\documentclass[aps, longbibliography, a4paper, nofootinbib, 11pt, onecolumn, notitlepage, tightenlines]{revtex4-1}
\usepackage{amsmath}
\usepackage{amssymb}
\usepackage[english]{babel}
\usepackage[utf8]{inputenc}
\usepackage[]{natbib}
\usepackage{hyperref}
\hypersetup{
	citecolor = red,
	colorlinks=true,
	linkcolor=cyan,
	filecolor=magenta,      
	urlcolor=blue,
}
\usepackage{enumerate}
\usepackage{graphicx}
\usepackage{bbm}
\usepackage{mathrsfs}
\usepackage{amsthm}
\newcommand\numberthis{\addtocounter{equation}{1}\tag{\theequation}}
\DeclareMathOperator{\R}{\mathbb{R}}
\DeclareMathOperator{\tr}{\mathrm{tr}}
\DeclareMathOperator{\id}{\mathbbm{1}}
\DeclareMathOperator{\im}{\mathbf{i}}
\newcommand\ket[1]{\left| #1 \right\rangle}
\newcommand\bra[1]{\left\langle #1 \right|}
\newcommand\braket[2]{\left\langle #1 |#2\right\rangle}

\newcommand{\Hil}{\mathcal{H}}
\newcommand{\Lin}{\mathcal{L}}
\newcommand{\Den}{\mathcal{D}}
\newcommand{\C}{\mathbb{C}}
\newcommand{\Op}{\mathrm{Op}}
\newcommand{\Ree}{\mathrm{Re}}
\newcommand{\ketbra}[2]{\ket{#1}\!\!\bra{#2}}
\newcommand{\supp}{\mathrm{supp}}
\newcommand{\CPTP}{\mathrm{CPTP}}
\newcommand{\Map}{\mathcal{M}}

\usepackage{hhline}
\newtheorem{theorem}{Theorem}
\newtheorem{lemma}[theorem]{Lemma}
\newtheorem{proposition}[theorem]{Proposition}
\newtheorem{corollary}[theorem]{Corrolary}
\newtheorem{definition}[theorem]{Defintion}

\AtBeginDocument{%
	\newwrite\bibnotes
	\def\bibnotesext{Notes.bib}
	\immediate\openout\bibnotes=\jobname\bibnotesext
	\immediate\write\bibnotes{@CONTROL{REVTEX41Control}}
	\immediate\write\bibnotes{@CONTROL{%
			apsrev41Control,author="08",editor="1",pages="1",title="0",year="1"}}
	\if@filesw
	\immediate\write\@auxout{\string\citation{apsrev41Control}}%
	\fi
}%

\begin{document}
\title{Decomposition Rules for Quantum R\'enyi Mutual Information\\with an Application to Information Exclusion Relations}

\author{Alexander McKinlay}
\author{Marco Tomamichel}
\affiliation{Centre for Quantum Software and Information,
		University of Technology Sydney,
		Ultimo NSW 2007, Australia}
\date{\today}

\begin{abstract}
	We prove decomposition rules for quantum R\'enyi mutual information, generalising the relation $I(A:B) = H(A) - H(A|B)$ to inequalities between R\'enyi mutual information and R\'enyi entropy of different orders. The proof uses Beigi's generalisation of Reisz-Thorin interpolation to operator norms~\cite{Bei}, and a variation of the argument employed by Dupuis~\cite{Dup} which was used to show chain rules for conditional R\'enyi entropies. The resulting decomposition rule is then applied to establish an information exclusion relation for R\'enyi mutual information, generalising the original relation by Hall~\cite{H95}.
\end{abstract}

\maketitle

\section{Introduction}

Mutual information is a fundamental quantity in information theory and can be interpreted as a measure of correlation between two random variables. Most prominently, Shannon~\cite{shannon48} established that the capacity of any discrete memoryless communication channel is given by the maximal mutual information between the channels input and output. Beyond its original use in information theory, it has found many other applications in information processing from such a wide range as machine learning (see, e.g.,~\cite{hirche18, ML1, ML2}) and computational linguistics (see, e.g.,~\cite{CL}). 
Moreover, quantum mutual information has analogous applications in quantum information, for example characterising the capacity of classical to quantum channels~\cite{holevo98,schumacher97,holevo73b} and the quantum channel capacity under entanglement assistance~\cite{Bennett,bennett02,bennett09,bertachristandl11}. It has also found applications in other areas of quantum physics, for example as an entanglement and correlation measure (see, e.g.,~\cite{brandao13d}) and to quantify Heisenberg's uncertainty principle (see~\cite{H95} and~\cite{coles17} for a review on related work). 

Quantum mutual information can be expressed in various equivalent ways, each of which elucidates different properties and interpretations of the quantity. Often quantum mutual information is defined in terms of the von Neumann entropy of its marginals (formal definitions and a discussion of properties of all the quantities mentioned here follow in Section~\ref{sec:relent}), namely we may write
\begin{equation}
	I(A:B) = H(A) + H(B) - H(AB) = H(A) - H(A|B) = H(B) - H(B|A) \,. \label{eq:def-via-decomposition}
\end{equation}
These relations can be interpreted as \emph{decomposition rules}, expressing the mutual information in terms of the\,---\,conceptually simpler\,---\,von Neumann entropies of the different marginals of the joint state of the systems $A$ and $B$. By appealing to the intuition that entropy measures uncertainty in a quantum system, they reveal that mutual information measures the uncertainty in $A$ that is due to the lack of knowledge of $B$, and vice versa. Another important and equivalent definition of mutual information is given in terms of Umegaki's relative entropy~\cite{umegaki62}, namely as the minimal relative entropy between the joint state $\rho_{AB}$ and any product state between the two systems, i.e.\
\begin{align}
	I(A:B)_\rho = \min_{\sigma_A,\sigma_B} D(\rho_{AB}\|\sigma_A\otimes \sigma_B) = \min_{\sigma_B} D(\rho_{AB}\|\rho_A\otimes \sigma_B) = D(\rho_{AB}\|\rho_A\otimes \rho_B) \,, \label{eq:def-via-divergence}
\end{align}
where in the second and third expressions we used that the minimum is taken for the marginals $\rho_A$ and $\rho_B$ of the joint state $\rho_{AB}$, respectively. This expression reveals a fundamental property of the quantum mutual information that is not evident from the decomposition rules, namely the \emph{data-processing inequality}. Specifically, this property entails that quantum mutual information is monotonically non-increasing under any local processing of information on $A$ and $B$. Its satisfaction directly follows from the monotonicity under quantum channels of the underlying relative entropy and the above equivalence. This property is crucial for many applications of the mutual information since it corresponds to our intuition that correlations cannot be created by acting on only one of the constituent parts (or by acting on them independently).

Following the footsteps of R\'enyi~\cite{renyi61}, various generalisations of the concept of mutual information to a one-parameter family (parametrised by the \emph{R\'enyi order} $\alpha$) of operationally significant measures have been proposed both in the classical~(see, e.g.,~\cite{verdu15,tomamichel17c}, for recent discussions) and the quantum setting~(see, e.g.,~\cite{hayashitomamichel15c}). We call such measures (quantum) R\'enyi mutual information if they satisfy the data-processing inequality.
Definitions that naturally satisfy the data-processing inequality are found by replacing the relative entropy in~Eq.~\eqref{eq:def-via-divergence} with a (quantum) R\'enyi divergence. In this work we will limit our attention to definitions based on minimal (or `sandwiched') R\'enyi divergence~\cite{lennert13,wilde13} as we plan to take advantage of their close relation to non-commutative norms. For example, we will consider the quantum R\'enyi mutual information (see Section~\ref{sec:sand} for formal definitions)
\begin{align}
	I_{\alpha}(A\;;\>\!B)_\rho := \min_{\sigma_B} {D}_{\alpha}(\rho_{AB} \| \rho_A \otimes \sigma_B) \,,
	\label{eq:def-q}
\end{align}
where ${D}_{\alpha}(\cdot\|\cdot)$ denotes the minimal R\'enyi divergence and $\alpha \in [\frac12,\infty)$. This generalises the quantum mutual information, which can be recovered by setting $\alpha = 1$. This and similar constructions of quantum R\'enyi mutual information have found direct operational interpretation in classical and quantum hypothesis testing~\cite{hayashitomamichel15c,tomamichel17c} and are widely used in the analysis of channel coding problems (see, e.g.,~\cite{wilde13,mosonyi11,mosonyi14-2}). 
It is important to note here that the equivalences in Eqs.~\eqref{eq:def-via-decomposition} and~\eqref{eq:def-via-divergence} no longer hold in the case of R\'enyi mutual information, and in particular if we start with Eq.~\eqref{eq:def-q} then we are lacking a way to decompose the R\'enyi mutual information into R\'enyi entropies of its marginals.

The main result of this paper is to fill this gap, in the sense of providing decomposition rules for (quantum) R\'enyi mutual information that generalise Eq.~\eqref{eq:def-via-decomposition}. These rules take the form
\begin{align}
	I_{\alpha}(A\;;\>\!B) \geq H_{\beta}(A) - H_{\gamma}(A|B)  \quad \textrm{and} 	\quad
	I_{\alpha}(A\;;\>\!B) \leq H_{\bar\beta}(A) - H_{\bar\gamma}(A|B)
\end{align}
for suitable choices of R\'enyi orders $\beta, \bar\beta, \gamma$ and $\bar\gamma$. The formal result is presented in Theorem~\ref{result1}. The two inequalities above reduce to the equality in Eq.~\eqref{eq:def-via-decomposition} when we take all the parameters to $1$.
The proof uses norm interpolation techniques~\cite{Bei}, and is inspired by the proof of similar decompositions in~\cite{Dup}, which take the form of chain rules for conditional R\'enyi entropies.

We explore an application of Theorem~\ref{result1} to information exclusion relations. These relations, dual to Heisenberg's uncertainty principle, give upper bounds on the total amount of correlation between a state measured in either one of two incompatible bases and some classical memory with information about how the initial state was prepared. Hall~\cite{H95} first formalised this as the bound
\begin{equation}\label{HallRel}
I(X:Y) + I(Z:Y) \leq \log(d^2c) \,.
\end{equation}
In the above, $X$ is the classical register produced by the measurement map on the $A$ subsystem in an orthonormal basis $\{\ket{e_x}\}_x$, and similarly for $Z$ and $\{\ket{f_z}\}_z$. Moreover, $d$ is the dimension of the system $A$ and $c = \max_{x,z}| \! \braket{e_x}{f_z} \!|^2$ is the maximal overlap of the two bases. Our second result is to give a parametrised family of information exclusion relations for R\'enyi mutual information. 
Indeed, Theorem~\ref{result2} not only generalises Eq.~\eqref{HallRel}, but also the following improvement involving quantum memory~\cite{coles14}:
\begin{equation}
I(X:B) + I(Z:B) \leq \log(d^2c) - H(A|B)
\end{equation}
Further applications, for example to quantum cryptography, and a potential generalisation to the smooth entropy framework~\cite{renner05b,mythesis}, are left as open questions.
 
The remainder of this paper is structured as follows. Section~\ref{sec:not} introduces the necessary notation and definitions. Section~\ref{sec:res} formally presents our two main theorems, with the proofs given in Sections~\ref{sec:proof1} and~\ref{sec:proof2}, respectively.

\section{Notation and definitions}
\label{sec:not}

We use standard notation as summarised in Table~\ref{nottable}. Additionally, shorthands are used for the following expressions:
\begin{align}
\alpha' = \frac{\alpha-1}{\alpha}\quad\text{and} \quad
\hat\alpha =\frac{\alpha}{2\alpha-1}, \text{ i.e. } \frac{1}{\alpha}+\frac{1}{\hat\alpha} = 2.
\end{align}
This produces some equivalences that will be useful for later calculations. We have
\begin{align}
\alpha'\alpha&=\alpha'\left(\frac{1}{1-\alpha'} \right)= \frac{\alpha'}{1-\alpha'},\\
\alpha'\hat\alpha &= \alpha'\left(\frac{\alpha}{2\alpha-1}\right)=\alpha'\alpha\left(\frac{2}{1-\alpha'}-1\right)^{-1} = \left(\frac{\alpha'}{1-\alpha'}\right)\left(\frac{1-\alpha'}{1+\alpha'}\right) = \frac{\alpha'}{1+\alpha'}\quad\text{and}\\
-\alpha' &= \frac{1-\alpha}{\alpha} = \left(1- \frac{\hat\alpha}{2\hat\alpha-1}\right)\left(\frac{2\hat\alpha-1}{\hat\alpha}\right) =\frac{2\hat\alpha-1}{\hat\alpha}-1 = \hat\alpha'.
\end{align}
\begin{table}\caption{Overview of notation}\label{nottable}
	\begin{tabular}{|c|l|}
		\hline
		Symbol&Meaning\\
		\hhline{|=|=|}
		$\log$& The logarithm to base 2\\
		\hline
		$A,B,C$& Quantum system or subsystems\\
		\hline
		$\Hil_A, \Hil_B, \Hil_C$& The Hilbert spaces corresponding to the quantum systems, $A$, $B$ and $C$.\\
		\hline
		$\Hil_{AB}$& $\Hil_A \otimes \Hil_B$\\
		\hline
		$\Lin(A, B)$& Set of linear operators from $\Hil_A$ to $\Hil_B$.\\
		\hline
		$\Lin(A)$& $\Lin(A, A)$\\
		\hline
		$\Den(A)$& The set of positive semi-definite operators in $\Lin(A)$ with unit trace.\\
		\hline
		$\CPTP(A,B)$&The set of completely-positive trace-preserving operator maps from $\Lin(A)$ to $\Lin(B)$\\
		\hline
		$\tr_A(\cdot)$& The partial trace over $A$, $\tr_A(\rho_A \otimes \rho_B) = \tr(\rho_A)\rho_B$.\\
		\hline
		$\rho_A$ & The marginal on $A$. For $\rho_{AB} \in \Den(AB)$, $\rho_A = \tr_B(\rho_{AB})$\\
		\hline
		$\sigma\gg\rho$& $\sigma$ `dominates' $\rho$, i.e the kernel of $\sigma$ is contained in the kernel of $\rho$.\\
		\hline
		$\sigma\perp\rho$& $\sigma$ and $\rho$ are perpendicular, i.e. the images of $\sigma$ and $\rho$ have empty intersection.\\$\sigma\not\perp\rho$&Not perpendicular\\
		\hline
		$\id_A\in \Lin(A)$ & The identity map on $\Hil_A$\\
		\hline
		$X_{A\rightarrow B}$& An operator in $\Lin(A,B)$\\
		\hline
		$\Op_{A\rightarrow B}(\cdot)$& $\Op_{A\rightarrow B}:\Hil_{AB} \rightarrow \Lin(A,B)$. For basis vectors $\ket{e_i}\in \Hil_A, \ket{f_j} \in \Hil_B$, $\ket{e_i}\otimes \ket{f_j}\mapsto \ketbra{f_j}{e_i}$\\
		\hline
		$\Op(\ket{\psi})$& The operator representation of $\ket{\psi}\in \Hil$\\
		\hline
		$\|\cdot\|_p$& The operator $p$-norm, $\|X\|_p=\tr\left[(X^\dagger X)^\frac{p}{2}\right]^\frac{1}{p}$. This is not a norm for $p<1$.\\
		\hline
	\end{tabular}
\end{table}

\subsection{R\'enyi entropy}
\label{sec:relent}

Originally proposed in~\cite{renyi61}, the R\'enyi entropy of order $\alpha\in (0,1) \cup (1, \infty)$ of a classical random variable $X$, distributed according to the probability law $P$, is defined as
\begin{equation}
H_\alpha(X) = \frac{1}{1-\alpha}\log\left(\sum_x P(X=x)^\alpha\right).
\end{equation}
It generalises the well-known Shannon entropy~\cite{shannon48} and serves to weigh outcomes with more or less likelihood differently depending on the order $\alpha$. The Shannon entropy is recovered in the limit $\alpha\rightarrow 1$.

The quantum R\'enyi entropy is a quantum generalisation of the R\'enyi entropy and is derived in an analogous way to von Neumann entropy: for a probability density matrix $\rho\in \Den(A)$, we define
\begin{equation}
H_\alpha(A)_\rho = \frac{1}{1-\alpha}\log\tr(\rho^\alpha).
\end{equation}
There are some particular choices of $\alpha$ which are either mathematically convenient or reflect specific physical situations.

When $\alpha\rightarrow1$ we recover the von Neumann entropy, which we denote
\begin{equation}
H_1(A)_\rho := H(A)_\rho = -\tr(\rho\log\rho).
\end{equation}
When $\alpha = 2$ we have the `collision' entropy, which characterises the purity of a quantum system:
\begin{equation}
H_2(A)_\rho = -\log\tr(\rho^2).
\end{equation}
The max-entropy could be naturally defined for $\alpha \rightarrow 0$ but, due to some mathematical restrictions, we instead use $\alpha = 1/2$. We have
\begin{align}
H_0(A)_\rho &= \log|\supp(\rho)|,\\
H_\frac{1}{2}(A)_\rho = H_\text{max}(A)_\rho &= 2\log\tr(\sqrt{\rho}).
\end{align}
The last quantity, and perhaps the most useful except for $\alpha \rightarrow 1$, is the min-entropy which we obtain for $\alpha\rightarrow \infty$.
\begin{align}
H_\infty(A)_\rho = H_\text{min}(A)_\rho = -\log\max_{i}\lambda_i,
\end{align}
where $\{\lambda_i\}$ are the eigenvalues of $\rho$.

\subsection{Minimal R\'enyi divergence and related quantities}
\label{sec:sand}

We now introduce the `sandwiched' R\'enyi divergence~\cite{lennert13, wilde13}. For $\rho, \sigma \in \Den(A)$ and ${\alpha \in (0,1) \cup (1,\infty)}$
\begin{equation}
D_\alpha(\rho\|\sigma):=\begin{cases}
\frac{1}{\alpha-1}\log\tr\left[\left({\sigma}^{\frac{1-\alpha}{2\alpha}}\rho{\sigma}^{\frac{1-\alpha}{2\alpha}}\right)^\alpha\right]\quad&\text{if } \rho \not\perp \sigma \wedge (\sigma\gg\rho \vee \alpha<1)\\
\infty \quad&\text{else}
\end{cases} .
\end{equation}
From this point we will refer to this quantity as simply `R\'enyi divergence'. Before we explore this concept further we first look at the simpler case where $\alpha = 1$. This case recovers Umegaki's relative entropy~\cite{umegaki62}, often called the quantum relative entropy,
\begin{equation}
D(\rho\|\sigma) = \tr\left[\rho(\log\rho - \log\sigma)\right] \,.
\end{equation}
We can use this quantity to obtain definitions of the von Neumann entropies that are equivalent to the intuitive definitions derived from the chain and decomposition rules, i.e.
\begin{align}
H(AB)_\rho &= -D(\rho_{AB}\|\id_{AB}),\\
H(A|B)_\rho &= -D(\rho_{AB}\|\id_A\otimes \rho_B)\\
&=-\tr(\rho_{AB}\log\rho_{AB})+\tr(\rho_B\log\rho_B)\\
&=H(AB)_\rho - H(B)_\rho,\\
I(A:B)_\rho&= D(\rho_{AB}\|\rho_A\otimes\rho_B)\\
&=\tr(\rho_{AB}\log\rho_{AB})-\tr(\rho_{A}\log\rho_{A})-\tr(\rho_{B}\log\rho_{B})\\
&=H(A)_\rho + H(B)_\rho - H(AB)_\rho\\
&=H(A)_\rho - H(A|B)_\rho.
\end{align}
In fact we obtain equivalent definitions for the mutual information by minimising over one or both subsystems in the following way. If we consider the positive-definiteness of the relative entropy due to Klein's inequality~\cite{Klein} and observe that when $\sigma_B = \rho_B$ then $D(\rho_B\|\sigma_B) = 0$ we can write
\begin{align}
I(A:B)_\rho &= D(\rho_{AB}\|\rho_A\otimes\rho_B) + \inf_{\sigma_B \in \Den(B)}D(\rho_B\|\sigma_B)\\
&=\tr(\rho_{AB}\log\rho_{AB} -\rho_{AB}\log(\rho_A\otimes\rho_B)) +\tr(\rho_B\log\rho_B) -\inf_{\sigma_B \in \Den(B)}\tr(\rho_B\log\sigma_B)\\
&= \inf_{\sigma_B \in \Den(B)}\tr(\rho_{AB}\log\rho_{AB} -\rho_{AB}\log(\rho_A\otimes\sigma_B))\\
&=\inf_{\sigma_B \in \Den(B)}D(\rho_{AB}\|\rho_A\otimes\sigma_B).
\end{align}
A similar calculation can be used to show that the equivalence also holds when minimised over both subsystems, i.e.
\begin{equation}
I(A:B)_\rho = \inf_{\substack{\sigma_A \in \Den(A)\\\sigma_B \in \Den(B)}}D(\rho_{AB}\|\sigma_A\otimes\sigma_B).
\end{equation}
Unfortunately, this equivalence does not extend to R\'enyi divergence but we can still define the relevant quantum R\'enyi entropies accordingly. The following notation is adapted from the notation introduced in~\cite{TBH}. We define the quantum R\'enyi entropy as
\begin{align}
H_\alpha(A)_\rho &= -D_\alpha(\rho_A\| \id_A).
\end{align}
The `sandwiched' conditional entropy can be defined
\begin{align}
\label{maxcond} H^\downarrow_\alpha(A|B)_\rho &= -D_\alpha(\rho_{AB}\|\id_A\otimes \rho_B)\quad\text{or}\\
H^\uparrow_\alpha(A|B)_\rho &= -\inf_{\sigma_B \in \Den(B)}D_\alpha(\rho_{AB}\|\id_A\otimes \sigma_B).
\end{align}
The ordering $H^\downarrow_\alpha(A|B)_\rho\leq H^\uparrow_\alpha(A|B)_\rho$, obvious from the definition, is indicated by the direction of the superscript arrow. We can safely assume that $\id_A\otimes \sigma_B\gg \rho_{AB}$, since any choice of $\sigma_B$ where this is not the case would certainly not achieve the infimum. The `sandwiched' mutual information~\cite{hayashitomamichel15c} is defined
\begin{align}
I^\uparrow_\alpha(A\;;\>\!B)_\rho &= \inf_{\sigma_B \in \Den(B)} D_\alpha(\rho_{AB}\|\rho_A\otimes \sigma_B)\label{maxmut}\quad\text{or}\\
\label{minmut} I^\downarrow_\alpha(A:B)_\rho &= \inf_{\substack{\sigma_A \in \Den(A)\\\sigma_B \in \Den(B)}} D_\alpha(\rho_{AB}\|\sigma_A\otimes \sigma_B).
\end{align}
Similarly, the superscript arrows indicate the ordering of each version and we satisfy the support condition as result of the minimisations. The use of `;' in Eq.~\eqref{maxmut} indicates that this quantity is not symmetric in its arguments, whereas Eq.~\eqref{minmut} is.

There are a few properties of the R\'enyi divergence which we find particularly useful: it generalises the the von Neumann entropy, it is monotone in $\alpha$~\cite{Bei}, it is mathematically convenient as it is closely related to norms and exhibits duality relations~\cite{Bei, lennert13}, and it satisfies the data-processing inequality~\cite{Bei, FL}.

These properties naturally extend to any quantity defined using R\'enyi divergence, which coincide with the mathematical and physical interpretation of quantum entropies. Monotonicity in $\alpha$ reflects the expected behaviour of R\'enyi entropy when weighing more or less likely outcomes differently. The data-processing inequality reflects that entropy can only ever increase (or correlation decrease) when information is processed (on each system independently).

We also make use of the following notation for the generalised R\'enyi mutual information~\cite{hayashitomamichel15c} and conditional entropy:
\begin{align}
H_\alpha(\rho_{AB}\|\tau_B) &= -D_\alpha(\rho_{AB}\|\id_A\otimes \tau_B),\label{GCE} \\
I_\alpha(\rho_{AB}\|\tau_A) &= \inf_{\sigma_B \in \Den(B)}D_\alpha(\rho_{AB}\|\tau_A\otimes \sigma_B).\label{GMI}
\end{align}
We can easily verify that the above quantities generalise Eqs.~\eqref{maxcond}-\eqref{minmut}, i.e.
\begin{align}
&&H^\downarrow_\alpha(A|B)_\rho &= H_\alpha(\rho_{AB}\|\rho_B),& H^\uparrow_\alpha(A|B)_\rho &= \sup_{\sigma_B \in \Den(B)}H_\alpha(\rho_{AB}\|\sigma_B),&&\\
&&I^\uparrow_\alpha(A\;;\>\!B)_\rho &= I_\alpha(\rho_{AB}\|\rho_A),& I^\downarrow_\alpha(A:B)_\rho &= \inf_{\sigma_A \in \Den(A)}I_\alpha(\rho_{AB}\|\sigma_A).&&
\end{align}

\section{Formal results and discussion}
\label{sec:res}

\subsection{Quantum R\'enyi mutual information decomposition rules}\label{cont1}
We can establish the following generalisations of the von Neumann mutual information decomposition rule in the form of inequalities whose direction depends on the choice of R\'enyi order of each entropic quantity.
\begin{theorem}\label{result1}
	For $\alpha>0$, $\beta, \gamma\geq1/2$ satisfying $\dfrac{\alpha}{\alpha-1} = \dfrac{\beta}{\beta-1}+\dfrac{\gamma}{\gamma-1}$
	we have, in the case when ${(\alpha-1)}{(\beta-1)}{(\gamma-1)}>0$,
	\begin{align}
	I^\uparrow_\gamma(A\;;\>\!B)_\rho&\geq H_\beta(B)_\rho-H_\alpha^\downarrow(B|A)_{\rho},\label{res1-1}\\
	I^\downarrow_\gamma(A:B)_\rho&\geq H_\beta(B)_\rho-H^\uparrow_\alpha(B|A)_{\rho}.\label{res1-2}
	\end{align}
	Otherwise, when $(\alpha-1)(\beta-1)(\gamma-1)<0$,
	\begin{align}
	I^\uparrow_\gamma(A\;;\>\!B)_\rho&\leq H_\beta(B)_\rho-H^\downarrow_\alpha(B|A)_{\rho},\label{res1-3}\\
	I^\downarrow_\gamma(A:B)_\rho&\leq H_\beta(B)_\rho-H^\uparrow_\alpha(B|A)_{\rho}.\label{app}
	\end{align}
\end{theorem}
This theorem follows from two components. The first is a re-expression of the relevant entropies to an operator norm form using operator-vector correspondence (see~\cite{Wat}), based on the technique used in~\cite{Dup}. The second is a Riesz-Thorin type interpolation result for operator norms developed in~\cite{Bei}.

In the case of Eqs.~\eqref{app} and~\eqref{res1-2}, we have symmetry in $A$ and $B$, giving us two additional inequalities with the systems swapped on the right-hand side.

Although there are many valid choices for the parameters, these choices are surprisingly limited. Fixing one parameter often leads to a restriction on the available regions for the other two, especially in the case when one parameter is greater than 2, see appendix~\ref{res} for more detail. When applied to Theorem~\ref{result2}, these restrictions result in a relatively weak statement where we must have $\alpha>\frac{2}{3}$. Whether these valid ranges can be improved is still an open question.

However, due the monotonicity in $\alpha$ of the R\'enyi divergence we have that the quantum R\'enyi condtional entropy and the quantum R\'enyi mutual information are non-increasing and non-decreasing respectively. This implies that Eqs.~\eqref{res1-1} and~\eqref{res1-2} additionally hold for $\frac{\alpha}{\alpha-1} \geq \frac{\beta}{\beta-1}+\frac{\gamma}{\gamma-1}$
and Eqs.~\eqref{res1-3} and~\eqref{app} additionally hold for $\frac{\alpha}{\alpha-1} \leq \frac{\beta}{\beta-1}+\frac{\gamma}{\gamma-1}$.

Incorporating the conditional entropy chain rule from~\cite{Dup} we can establish the following supplementary inequalities.
\begin{corollary}\label{noncond}
	For $\alpha, \beta, \gamma, \delta>1/2$ such that $\dfrac{\delta}{\delta-1} = \dfrac{\alpha}{\alpha-1} + \dfrac{\beta}{\beta-1}+ \dfrac{\gamma}{\gamma-1}$ we have, when $\delta<\alpha, \beta, \gamma$
	\begin{align}
	I^\downarrow_\gamma(A:B)_\rho &\geq H_\alpha(A)_\rho + H_\beta(B)_\rho - H_\delta(AB)_\rho,
	\end{align}
	and, when $\delta>\alpha, \beta, \gamma$
	\begin{align}
	I^\uparrow_\gamma(A\;;\>\!B)_\rho &\leq H_\alpha(A)_\rho + H_\beta(B)_\rho - H_\delta(AB)_\rho.\label{reverse}
	\end{align}
\end{corollary}

\subsection{A quantum R\'enyi information exclusion relation}
As an immediate application of Theorem~\ref{result1}, we can derive the following bipartite information exclusion relation for quantum R\'enyi entropies.
\begin{theorem}\label{result2}
	For $\alpha>2/3$, $1/2 \leq \beta, \gamma<4/3$ satisfying both $\dfrac{\alpha}{\alpha-1} \leq \dfrac{1}{\beta-1}+\dfrac{1}{\gamma-1}$ and $({\alpha-1}) ({\beta-1}) ({\gamma-1}){<0}$,
	\begin{equation}\label{MICR}
	I^\uparrow_\beta(B\;;\>\!X)_\rho + I^\uparrow_\gamma(B\;;\>\!Z)_\rho \leq \log (d^2c)-H^\downarrow_\alpha(A|B)_\rho.
	\end{equation}
\end{theorem}
This follows from an application of Eq.~\eqref{app} to a generalisation of a Maassen-Uffink type relation developed in~\cite{Fin} using a method analogous to the derivation of the Hall principle Eq.~\eqref{HallRel}.

Given that $B$ is some classical memory, the Hall principle can be recovered when $\alpha, \beta, \gamma \rightarrow 1$.

We also have the following interesting cases: Choosing $\alpha = 2/3$ implies $\gamma = \frac{\beta}{2\beta -1}$, which in turn gives us $\beta = 1/2 \implies \gamma \rightarrow \infty$.

With $\alpha \rightarrow \infty$ we obtain $\gamma = \frac{2\beta - 3}{\beta -2}$ and $\beta = 1/2 \implies \gamma = 4/3$. This case is further explored in Corollary~\ref{res2c}.

For $\alpha = 2$ we have $\gamma = \frac{3\beta -4}{2\beta-3}$ and $\beta = 1/2 \implies \gamma = 5/4$.
%

\section{Proof of Theorem~\ref{result1}}
\label{sec:proof1}

The proof of Theorem~\ref{result1} draws on two major components: an interpolation result for operator norms (Theorem~\ref{Bei}) and a re-expression of the relevant entropic quantities to operator norms (Lemma~\ref{threeEnt}) which the interpolation can then be performed on. This machinery is then employed in Propositions~\ref{part1},~\ref{part2} and~\ref{part3}, which explore the possible permutations of parameters.
\subsection{Expressing the entropic quantities as operator norms}
We employ the following interpolation result from~\cite{Bei} which generalises Riesz-Thorin interpolation to operator norms. This result relies on H\"olders inequality and the log convexity found in the Hadamard three-line theorem. We present it in a slightly less general form than originally proposed. 
\begin{theorem}[Beigi~\cite{Bei}]\label{Bei}
	Let $F:S\rightarrow \Lin(A)$, be a bounded map from the complex strip $S:=\{{z\in\C}:0\leq\Ree(z)\leq1\}$ into the linear operators on $\Hil_A$ which is holomorphic on the interior of $S$ and continuous on the boundary. Let $0\leq \theta\leq 1$.
	\begin{align}
	&\text{For}&&\frac{1}{p_\theta} = \frac{1-\theta}{p_0} + \frac{\theta}{p_1},&&\\
	&\text{and with}&&M_k = \sup_{t\in \R}\|F(k+it)\|_{p_k},&&\\
	&\text{we have}&&\|F(\theta)\|_{p_\theta} \leq M_0^{1-\theta}M_1^{\theta}.&&
	\end{align}
\end{theorem}
The following lemma yields some useful operator norm forms of the relevant entropic quantities that are compatible with Theorem~\ref{Bei}.
\begin{lemma}\label{threeEnt}
For a pure state $\ket{\varphi} \in \Hil_{ABC}$ with $\ketbra{\varphi}{\varphi} = \rho$ and ${X = X_{B\rightarrow AC} = \Op_{B\rightarrow AC}(\ket{\varphi})}$. Given $(\sigma_A\otimes\sigma_B) \gg \rho_{AB}$ and $\alpha\in (0, 1)\cup (1, \infty)$ we have 
\begin{align}
H_\alpha(\rho_{AB}\|\sigma_A) &= -\log \sup_{\tau_{C}\in \Den(C)}\left\|\left(\sigma_A^{-1}\otimes\tau_C\right)^\frac{\alpha'}{2}X\right\|_{2}^\frac{2}{\alpha'},\label{CRE}\\
H_\alpha(B)_\rho &= -\log\left\|X\right\|_{2\alpha}^\frac{2}{\alpha'},\label{RE}
\end{align}
if in addition $\alpha\geq \frac{1}{2}$,
\begin{align}
I_\alpha(\rho_{AB}\|\sigma_B) &=\log\sup_{\tau_C\in \Den(C)}\left\|\left(\sigma_A^{-1}\otimes\tau_C\right)^\frac{\alpha'}{2}X\right\|_{2\hat\alpha}^\frac{2}{\alpha'}.\label{MI}
\end{align}
\end{lemma}
The proof of Lemma~\ref{threeEnt} relies on the operator-vector correspondence. Below, we summarise the relevant properties which follow from the definitions in Table~\ref{nottable}. For proofs see \cite[Chap. 1.1]{Wat}.
\begin{lemma}\label{OV1}
	Let $\ket{\psi} \in \Hil_{AB}$, $\rho_A\in \Lin(A)$ and $\sigma_B\in \Lin(B)$. Then $\Op_{A\rightarrow B}\left(\rho_A\otimes \sigma_B \ket{\psi}\right) = \sigma_B \Op_{A\rightarrow B}(\ket{\psi})\rho_A$.
\end{lemma}
\begin{lemma}\label{Isom}
	Let $\ket{\psi}\in\Hil_{AB}$. Then
	\begin{equation}
	\|\ket{\psi}\|_2 = \braket{\psi}{\psi} = \tr\left[\Op_{A\rightarrow B}(\ket{\psi})^\dagger \Op_{A\rightarrow B}(\ket{\psi})\right] =\|\Op_{A\rightarrow B}(\ket{\psi})\|_2.
	\end{equation}
\end{lemma}
\begin{lemma}\label{OV3}
	Let $\ket{\psi}\in \Hil_{AB}$, $\rho = \ketbra{\psi}{\psi}$. Then
	$\rho_A = \Op_{A\rightarrow B}(\ket{\psi})^\dagger \Op_{A\rightarrow B}(\ket{\psi})$
	and\\$
	{\rho_B = \Op_{A\rightarrow B}(\ket{\psi})\Op_{A\rightarrow B}(\ket{\psi})^\dagger}
	$.
\end{lemma}

\begin{proof}[Proof of Eq.~\eqref{CRE}]
We have from equation (19) in~\cite{lennert13} that we can write
\begin{align}
H_\alpha(\rho_{AB}\|\sigma_A) = -\log\sup_{\tau_C \in \Den(C)}\bra{\varphi}\sigma_A^{-\alpha'}\otimes\id_B\otimes\tau_C^{\alpha'}\ket{\varphi}^\frac{1}{\alpha'}.
\end{align}
Also, using Lemma~\ref{OV1} we have
\begin{align}
\Op_{B\rightarrow AC}\left(\sigma_A^\frac{-\alpha'}{2}\otimes\id_B\otimes\tau_C^\frac{\alpha'}{2}\ket{\varphi}\right) = \left(\sigma_A^\frac{-\alpha'}{2}\otimes\tau_C^\frac{\alpha'}{2}\right)X.
\end{align}
From this we can deduce, using Lemma~\ref{Isom},
\begin{align}
\bra{\varphi}\sigma_A^{-\alpha'}\otimes\id_B\otimes\tau_C^{\alpha'}\ket{\varphi}^\frac{1}{\alpha'} &= \left\|\sigma_A^\frac{-\alpha'}{2}\otimes\id_B\otimes\tau_C^\frac{\alpha'}{2}\ket{\varphi}\right\|_2^\frac{2}{\alpha'}\\
&= \left\|\Op_{B\rightarrow AC}\left(\sigma_A^\frac{-\alpha'}{2}\otimes\tau_C^\frac{\alpha'}{2}\ket{\varphi}\right)\right\|_2^\frac{2}{\alpha'}\\
&= \left\|\left(\sigma_A^\frac{-\alpha'}{2}\otimes\tau_C^\frac{\alpha'}{2}\right)X\right\|_2^\frac{2}{\alpha'}.
\end{align}
\end{proof}
\begin{proof}[Proof of Eq.~\eqref{RE}]
From Lemma~\ref{OV3} we can see that $\rho_B= X^\dagger X$, hence we have
\begin{align}
H_\alpha(B)_{\rho} = -\log\left\|X^\dagger X\right\|_\alpha^\frac{1}{\alpha'}= -\log\left\|X\right\|_{2\alpha}^\frac{2}{\alpha'}.
\end{align}
\end{proof}
The proof of $\eqref{MI}$ relies on a duality result from~\cite{hayashitomamichel15c}.
\begin{lemma}\label{Dual} For $\alpha\in [1/2, \infty)$, we have $I_\alpha(\rho_{AB}\|\tau_A) = - I_{\hat\alpha}(\rho_{AC}\|\tau_A^{-1})$.
\end{lemma}
\begin{proof}[Proof of Eq.~\eqref{MI}]
Using operator-vector correspondence we can re-express the generalised R\'enyi mutual information using operator norms.
\begin{align}
I_\alpha\left(\rho_{AB}\|\tau_A\right) &= \frac{1}{\alpha-1}\log\inf_{\sigma_B\in \Den(A)}\tr\left(\left[\left(\tau_A\otimes\sigma_B\right)^\frac{-\alpha'}{2}\rho_{AB}\left(\tau_A\otimes\sigma_B\right)^\frac{-\alpha'}{2}\right]^\alpha\right)\\
&=\log\inf_{\sigma_B\in \Den(A)}\left\|\left(\tau_A\otimes\sigma_B\right)^\frac{-\alpha'}{2}X_{C\rightarrow AB}X_{C\rightarrow AB}^\dagger\left(\tau_A\otimes\sigma_B\right)^\frac{-\alpha'}{2}\right\|_\alpha^\frac{1}{\alpha'}\\
&=\log\inf_{\sigma_B\in \Den(A)}\left\|\left(\tau_A\otimes\sigma_B\right)^\frac{-\alpha'}{2}X_{C\rightarrow AB}\right\|_{2\alpha}^\frac{2}{\alpha'}.\numberthis\label{MINorm}
\end{align}
Using Lemma~\ref{Dual} and Eq.~\eqref{MINorm} we can write
\begin{align}
I_\alpha\left(\rho_{AB}\|\sigma_A\right) &=-I_{\hat\alpha}\left(\rho_{AC}\|\sigma_A^{-1}\right)\\
&=-\log\inf_{\omega_C\in \Den(C)}\left\|\left(\sigma_A^\frac{\hat\alpha'}{2}\otimes\omega_C^\frac{-\hat\alpha'}{2}\right)X\right\|_{2\hat\alpha}^\frac{2}{\hat\alpha'}\\
&=\log\sup_{\omega_C\in \Den(C)}\left\|\left(\sigma_A^\frac{-\alpha'}{2}\otimes\omega_C^\frac{\alpha'}{2}\right)X\right\|_{2\hat\alpha}^\frac{2}{\alpha'}.
\end{align}
\end{proof}

\subsection{Applying Beigi's Theorem}
Before moving forward with the main component of the proof of Theorem~\ref{result1} we will first look at some motivation for the choice of parameters.

We want to use the interpolation result to find inequalities of the form
\[
- H_\alpha(B|A)_\rho \leq I_\gamma(A:B)_\rho - H_\beta(B)_\rho.
\]

Exponentiating on both sides and keeping in mind that we can express the resulting quantities as operator norms to the power of a function of the relevant parameter we obtain an inequality of the form
\begin{equation}\label{genform}
\|X_{B|A}\|_{p_\alpha}^\frac{2}{\alpha'}\leq \|X_{B}\|_{p_\beta}^\frac{2}{\beta'}\|X_{A:B}\|_{p_\gamma}^\frac{2}{\gamma'},
\end{equation}
where $X_{B|A}$, etc. are simply place-holders for the actual operators, used for brevity. We can then put Eq.~\eqref{genform} in the form required for Beigi's Theorem by taking both sides to the power of $\frac{\alpha'}{2}$, resulting in
\[
\|X_{B|A}\|_{p_\alpha}\leq \|X_{B}\|_{p_\beta}^\frac{\alpha'}{\beta'}\|X_{A:B}\|_{p_\gamma}^\frac{\alpha'}{\gamma'},
\]
where $1-\theta = \frac{\alpha'}{\beta'}$ and $\theta = \frac{\alpha'}{\gamma'}$. This implies
\begin{align}
\label{prerel}1-\frac{\alpha'}{\gamma'} = \frac{\alpha'}{\beta'}
\implies \frac{1}{\alpha'} = \frac{1}{\beta'} + \frac{1}{\gamma'}.\numberthis
\end{align}

We can find the reverse of the inequality in Eq.~\eqref{genform} by negating all the exponents but this does not affect Eq.~\eqref{prerel}. Additionally, the order of the quantities in Eq.~\eqref{genform} has no effect, since we can choose a $\theta$ in each case that reproduces Eq.~\eqref{prerel}.

For example we could rewrite Eq.~\eqref{genform} as
\begin{equation}
\|X_{A:B}\|_{p_\gamma}^\frac{-2}{\gamma'}\leq \|X_{B|A}\|_{p_\alpha}^\frac{-2}{\alpha'}\|X_{B}\|_{p_\beta}^\frac{2}{\beta'}.
\end{equation}
To apply Theorem~\ref{Bei} in this case we would choose $1-\theta = \frac{\gamma'}{\alpha'}$ and $\theta = \frac{-\gamma'}{\beta'}$, resulting in $1+\frac{\gamma'}{\beta'} = \frac{\gamma'}{\alpha'}$,
which is again Eq.~\eqref{prerel}.

A more in-depth discussion of the implications and restrictions of this condition, which inform the choices in the following propositions, is deferred to Appendix~\ref{res}.

Theorem~\ref{result1} can be proved directly from the following propositions which make use of the above results.
\begin{proposition}\label{part1}
Let $\alpha,\beta,\gamma$ be such that $\frac{1}{\alpha'} = \frac{1}{\beta'}+\frac{1}{\gamma'}$. Then, the following holds:\\

For $\alpha\in (1,2),\beta,\gamma\in(1,\infty)$, we find
\begin{align}
I^\uparrow_\gamma(A\;;\>\!B)_\rho&\geq H_\beta(B)_\rho-H^\downarrow_\alpha(B|A)_{\rho} \text{ and}\label{p1ugd}\\
I^\downarrow_\gamma(A:B)_\rho&\geq H_\beta(B)_\rho-H^\uparrow_\alpha(B|A)_{\rho}\label{p1dgu}.
\end{align}

For $\alpha\in \left[2/3,1\right),\beta,\gamma\in\left[1/2,1\right)$, we find
\begin{align}
I^\uparrow_\gamma(A\;;\>\!B)_\rho&\leq H_\beta(B)_\rho-H^\downarrow_\alpha(B|A)_{\rho}\label{p1uld} \text{ and}\\
I^\downarrow_\gamma(A:B)_\rho&\leq H_\beta(B)_\rho-H^\uparrow_\alpha(B|A)_{\rho}\label{p1dlu}.
\end{align}
\end{proposition}
\begin{proof}
Choose $F(z) = (\sigma_A^{-1}\otimes\tau_C)^\frac{z\gamma'}{2}X,\quad
\theta = \frac{\alpha'}{\gamma'},\quad
p_0 = 2\beta,\quad
p_1 = 2\hat\gamma.$
With these choices we can determine $\theta = \alpha'\left(\frac{1}{\alpha'}-\frac{1}{\beta'} \right) = 1 - \frac{\alpha'}{\beta'}$, hence $1-\theta = \frac{\alpha'}{\beta'}$.

We can also calculate the appropriate value of $p_\theta$ to use Theorem~\ref{Bei}:
\begin{align}
\frac{1}{p_\theta} = \frac{\alpha'}{2\beta'\beta} + \frac{\alpha'}{2\gamma'\hat\gamma}\implies\frac{2}{\alpha'p_\theta} =\frac{1-\beta'}{\beta'} + \frac{1+\gamma'}{\gamma'}=\frac{\gamma'+\beta'}{\beta'\gamma'}= \frac{1}{\beta'} + \frac{1}{\gamma'},
\end{align}
thus we can conclude that $p_\theta = 2$.

We can therefore calculate that
\begin{equation}
\left\|F(\theta)\right\|_{p_\theta} = \left\|(\sigma_A^{-1}\otimes\tau_C)^\frac{\alpha'}{2}X\right\|_2.
\end{equation}
Additionally,
\begin{equation}
\|F(\im t)\|_{p_0} = \left\|(\sigma_A^{-1}\otimes\tau_C)^\frac{\im t\gamma'}{2}X\right\|_{2\beta}
\end{equation}
and
\begin{equation}
\|F(1+\im t)\|_{p_1} = \left\|(\sigma_A^{-1}\otimes\tau_C)^\frac{(1+\im t)\gamma'}{2}X\right\|_{2\hat\gamma}.
\end{equation}
Since $(\sigma_A^{-1}\otimes\tau_C)^\frac{\im t\gamma'}{2}$ is unitary for all $t\in \R$ we can write
\begin{equation}
M_0 = \left\|X\right\|_{2\beta}\quad\text{and}\quad M_1 = \left\|(\sigma_A^{-1}\otimes\tau_C)^\frac{\gamma'}{2}X\right\|_{2\hat\gamma}.
\end{equation}
Applying Theorem~\ref{Bei} we have
\begin{align}
\left\|(\sigma_A^{-1}\otimes\tau_C)^\frac{\alpha'}{2}X\right\|_2&\leq \left\|X\right\|_{2\beta}^\frac{\alpha'}{\beta'}\left\|(\sigma_A^{-1}\otimes\tau_C)^\frac{\gamma'}{2}X\right\|_{2\hat\gamma}^\frac{\alpha'}{\gamma'}.
\end{align}

First, consider $\alpha'>0$. Maximising over $\tau_C$ on both sides we have
\begin{align}
\sup_{\tau_{C}\in \Den(C)}\left\|(\sigma_A^{-1}\otimes\tau_C)^\frac{\alpha'}{2}X\right\|_2^\frac{2}{\alpha'}&\leq \left\|X\right\|_{2\beta}^\frac{2}{\beta'}\sup_{\tau_{C}\in \Den(C)}\left\|(\sigma_A^{-1}\otimes\tau_C)^\frac{\gamma'}{2}X\right\|_{2\hat\gamma}^\frac{2}{\gamma'}.
\end{align}
Choose $\sigma_A = \rho_A$. Then
\begin{align}
\sup_{\tau_{C}\in \Den(C)}\left\|(\rho_A^{-1}\otimes\tau_C)^\frac{\alpha'}{2}X\right\|_2^\frac{2}{\alpha'}&\leq \sup_{\tau_{C}\in \Den(C)}\left\|X\right\|_{2\beta}^\frac{2}{\beta'}\left\|(\rho_A^{-1}\otimes\tau_C)^\frac{\gamma'}{2}X\right\|_{2\hat\gamma}^\frac{2}{\gamma'}.
\end{align}
Using Lemma~\ref{threeEnt}, we can rewrite this as
\begin{align}
-H^\downarrow_\alpha(B|A)_{\rho} &\leq -H_\beta(B)_\rho + I^\uparrow_\gamma(A\;;\>\!B)_\rho\\
\implies I^\uparrow_\gamma(A\;;\>\!B)_\rho&\geq H_\beta(B)_\rho-H^\downarrow_\alpha(B|A)_{\rho}.
\end{align}
Similarly, if we minimise over $\sigma_A$ on both sides we arrive at Eq.~\eqref{p1dgu}.

If instead $\alpha'<0$, we obtain
\begin{align}
\sup_{\tau_{C}\in \Den(C)}\left\|(\sigma_A^{-1}\otimes\tau_C)^\frac{\alpha'}{2}X\right\|_2^\frac{2}{\alpha'}&\geq \sup_{\tau_{C}\in \Den(C)}\left\|X\right\|_{2\beta}^\frac{2}{\beta'}\left\|(\sigma_A^{-1}\otimes\tau_C)^\frac{\gamma'}{2}X\right\|_{2\hat\gamma}^\frac{2}{\gamma'}.
\end{align}
We can again choose $\sigma_A = \rho_A$ or minimise over $\sigma_A$, giving us Eqs.~\eqref{p1uld} and~\eqref{p1dlu} respectively. The valid ranges can be determined using Lemma~\ref{abylem}.
\end{proof}
\begin{proposition}\label{part2}
Let $\alpha,\beta,\gamma$ be such that $\frac{1}{\alpha'} = \frac{1}{\beta'}+\frac{1}{\gamma'}$. Then, the following holds.\\

For $\alpha\in(0,1), \gamma \in [1/2,1),\beta\in(1,\infty)$, we find
\begin{align}
I^\uparrow_\gamma(A\;;\>\!B)_\rho&\geq H_\beta(B)_\rho-H^\downarrow_\alpha(B|A)_{\rho}\label{p2ugd} \text{ and}\\
I^\downarrow_\gamma(A:B)_\rho&\geq H_\beta(B)_\rho-H^\uparrow_\alpha(B|A)_{\rho}\label{p2dgu}.
\end{align}

For $\beta\in[1/2,1), \gamma \in (1,2),\alpha\in(1,\infty)$, we find
\begin{align}
I^\uparrow_\gamma(A\;;\>\!B)_\rho&\leq H_\beta(B)_\rho-H^\downarrow_\alpha(B|A)_{\rho}\label{p2uld} \text{ and}\\
I^\downarrow_\gamma(A:B)_\rho&\leq H_\beta(B)_\rho-H^\uparrow_\alpha(B|A)_{\rho}\label{p2dlu}.
\end{align}
\end{proposition}
\begin{proof}
Choose $F(z) = (\sigma_A^{-1}\otimes\tau_C)^\frac{z\alpha'}{2}X,\quad
\theta = \frac{\gamma'}{\alpha'},\quad
p_0 = 2\beta,\quad
p_1 = 2$.
We have, as before, $1-\theta = \frac{-\gamma'}{\beta'}$ and through a similar calculation we can conclude that $p_\theta= 2\hat\gamma$.

We have
\begin{align}
&\|F(\theta)\|_{p_\theta} = \left\|(\sigma_A^{-1}\otimes\tau_C)^\frac{\gamma'}{2}X\right\|_{2\hat\gamma}, \quad
\|F(\im t)\|_{p_0} = \left\|(\sigma_A^{-1}\otimes\tau_C)^\frac{\im t\alpha'}{2}X\right\|_{2\beta}\\
&\text{and }\|F(1+\im t)\|_{p_1} = \left\|(\sigma_A^{-1}\otimes\tau_C)^\frac{(1+\im t)\alpha'}{2}X\right\|_{2},
\end{align}
hence $M_0 = \left\|X\right\|_{2\beta}\text{ and } M_1 = \left\|(\sigma_A^{-1}\otimes\tau_C)^\frac{\alpha'}{2}X\right\|_{2}$.
Applying Theorem~\ref{Bei} we have
\begin{align}
\left\|(\sigma_A^{-1}\otimes\tau_C)^\frac{\gamma'}{2}X\right\|_{2\hat\gamma}&\leq\left\|X\right\|_{2\beta}^\frac{-\gamma'}{\beta'}\left\|(\sigma_A^{-1}\otimes\tau_C)^\frac{\alpha'}{2}X\right\|_{2}^\frac{\gamma'}{\alpha'}.
\end{align}

First, we consider the case where $\gamma'>0$. It follows that
\begin{align} \left\|(\sigma_A^{-1}\otimes\tau_C)^\frac{\gamma'}{2}X\right\|_{2\hat\gamma}^\frac{2}{\gamma'}&\leq\left\|X\right\|_{2\beta}^\frac{-2}{\beta'}\left\|(\sigma_A^{-1}\otimes\tau_C)^\frac{\alpha'}{2}X\right\|_{2}^\frac{2}{\alpha'}.
\end{align}

As in Proposition~\ref{part1}, we can maximise over $\tau_C$ and on both sides. Continuing the same procedure by choosing $\sigma_A = \rho_A$ or minimising over $\sigma_A$ we arrive at Eqs.~\eqref{p2uld} and~\eqref{p2dlu}.
Repeating the same process with the assumption $\gamma'<0$ yields Eqs.~\eqref{p2ugd} and~\eqref{p2dgu}. We can again refer to Lemma~\ref{abylem} to determine the valid ranges.
\end{proof}
\begin{proposition}\label{part3}Let $\alpha,\beta,\gamma$ be such that $\frac{1}{\alpha'} = \frac{1}{\beta'}+\frac{1}{\gamma'}$. Then, the following holds.\\
	
	For $\alpha\in(0,1), \beta \in [1/2,1),\gamma\in(1,\infty)$, we find
	\begin{align}
	I^\uparrow_\gamma(A\;;\>\!B)_\rho&\geq H_\beta(B)_\rho-H^\downarrow_\alpha(B|A)_{\rho}\label{p3ugd} \text{ and}\\
	I^\downarrow_\gamma(A:B)_\rho&\geq H_\beta(B)_\rho-H^\uparrow_\alpha(B|A)_{\rho}\label{p3dgu}.
	\end{align}
	
	For $\gamma\in[1/2,1), \beta \in (1,2),\alpha\in(1,\infty)$, we find
	\begin{align}
	I^\uparrow_\gamma(A\;;\>\!B)_\rho&\leq H_\beta(B)_\rho-H^\downarrow_\alpha(B|A)_{\rho}\label{p3uld} \text{ and}\\
	I^\downarrow_\gamma(A:B)_\rho&\leq H_\beta(B)_\rho-H^\uparrow_\alpha(B|A)_{\rho}\label{p3dlu}.
	\end{align}
\end{proposition}
\begin{proof}
Choose $F(z) = (\sigma_A^{-1}\otimes\tau_C)^{\frac{\gamma'}{2}-z\frac{\gamma'\alpha'}{2\beta'}}X,\quad
\theta = \frac{\beta'}{\alpha'},\quad
p_0 = 2\hat\gamma,\quad
p_1 = 2$.
As above, $1-\theta = \frac{-\beta'}{\gamma'}$ and $p_\theta= 2\beta$.

We have
\begin{align}
&\|F(\theta)\|_{p_\theta} = \left\|X\right\|_{2\beta},\quad
\|F(\im t)\|_{p_0} = \left\|(\sigma_A^{-1}\otimes\tau_C)^{\frac{\gamma'}{2}-\frac{\im t\gamma'\alpha'}{2\beta'}}X\right\|_{2\hat\gamma}\\
&\text{and }\|F(1+\im t)\|_{p_1} = \left\|(\sigma_A^{-1}\otimes\tau_C)^{\frac{\alpha'}{2}-\frac{\im t\gamma'\alpha'}{2\beta'}}X\right\|_{2},
\end{align}
hence $M_0 = \left\|(\sigma_A^{-1}\otimes\tau_C)^{\frac{\gamma'}{2}}X\right\|_{2\hat\gamma}\text{ and } M_1 = \left\|(\sigma_A^{-1}\otimes\tau_C)^\frac{\alpha'}{2}X\right\|_{2}$.
Applying Theorem~\ref{Bei} and performing the same procedure as in Propositions~\ref{part1} and~\ref{part2}, for both $\beta'>0$ and $\beta'<0$ we obtain Eqs.~\eqref{p3ugd},~\eqref{p3dgu},~\eqref{p3uld} and~\eqref{p3dlu}. For the valid ranges, we have a similar situation as in Proposition~\ref{part2} but with symmetry in $\beta$ and $\gamma$.
\end{proof}
We may now prove Theorem~\ref{result1}:
\begin{proof}[Proof of Theorem~\ref{result1}]
	All that remains is to combine the three propositions and examine the valid ranges.
	We have from Lemma~\ref{abylem} that the three propositions cover all possible permutations of the parameters, and hence all valid values of $\alpha, \beta$ and $\gamma$.
	
	For the forward inequality, i.e. Eqs.~\eqref{p1ugd},~\eqref{p1dgu},~\eqref{p2ugd},~\eqref{p2dgu},~\eqref{p3ugd} and~\eqref{p3dgu} we can see that either ($\alpha, \beta, \gamma > 1$), ($\alpha, \gamma<1, \beta > 1$) or ($\alpha, \beta <1, \gamma>1$), which all satisfy $(\alpha-1)(\beta-1)(\gamma-1)>0$.
	
	For the reverse inequality, i.e. Eqs.~\eqref{p1uld},~\eqref{p1dlu},~\eqref{p2uld},~\eqref{p2dlu},~\eqref{p3uld} and~\eqref{p3dlu} we have either (${\alpha, \beta, \gamma < 1}$), ($\alpha, \gamma >1, \beta > 1$) or ($\alpha, \beta>1, \gamma<1$), which all satisfy $(\alpha-1)(\beta-1)(\gamma-1)<0$.
\end{proof}

\subsection{Decomposition rule in terms of the joint entropy}

We now include the proof of Corollary~\ref{noncond}, showing that we may also establish a somewhat weaker inequality that does not involve the conditional entropy and generalises the alternative form of the quantum mutual information decomposition rule. Note that this alternative form is equivalent for $\alpha = 1$ but this equivalence does not extend to general R\'enyi order.
\begin{proof}[Proof of Corollary~\ref{noncond}]
	From Theorem 1 in~\cite{Dup} we have for $\frac{1}{\alpha'} = \frac{1}{\beta'}+\frac{1}{\gamma'}$ that
	\begin{align}
	H^\uparrow_\beta(A|B)_\rho\leq H_\alpha(AB)_\rho - H_\gamma(B)_\rho,\label{cr1}
	\end{align}
	if $(\alpha-1)(\beta-1)(\gamma-1)>0$, and
	\begin{align}
	H^\downarrow_\beta(B|A)_\rho\geq H_\alpha(AB)_\rho - H_\gamma(A)_\rho\label{cr2},
	\end{align}
	if $(\alpha-1)(\beta-1)(\gamma-1)<0$.\\
	
	We begin with Eq.~\eqref{res1-2} then substitute in Eq.~\eqref{cr1} with valid parameters
	\begin{align}
	I^\downarrow_\gamma(A:B)_\rho&\geq H_\alpha(A)_\rho-H^\uparrow_{\alpha_1}(A|B)_{\rho}\\
	&\geq H_\alpha(A)_\rho+ H_\beta(B)_\rho - H_\delta(AB)_\rho,
	\end{align}
	where $\frac{1}{\delta'} - \frac{1}{\beta'}= \frac{1}{\alpha'} + \frac{1}{\gamma'}$.
	We know from Corollary~\ref{shyeah} that both
	\begin{align}
	(\alpha-1)(\alpha_1-1)(\gamma-1)>0\quad\text{then}\quad&\alpha_1<\alpha, \gamma \quad\text{and}\\
	(\beta-1)(\alpha_1-1)(\delta-1)>0\quad\text{then}\quad&\delta<\alpha_1, \beta.
	\end{align}
	Similarly, if we begin with Eqs.~\eqref{res1-3} and~substitute in Eq.~\eqref{cr2} we arrive at Eq.~\eqref{reverse} but with
	\begin{align}
	(\alpha-1)(\alpha_1-1)(\gamma-1)<0\quad\text{then}\quad&\alpha_1>\alpha, \gamma,\quad\text{and}\\
	(\beta-1)(\alpha_1-1)(\delta-1)<0\quad\text{then}\quad&\delta>\alpha_1, \beta.
	\end{align}
\end{proof}

\section{Proof of Theorem~\ref{result2}}
\label{sec:proof2}

Hall's result~\cite{H95} follows from an extension of the Maassen-Uffink relation \citep{MU}, found in~\cite{coles17}:
\begin{equation}
H(X|Y) + H(Z|Y) \geq -\log c.
\end{equation}
Substituting the Shannon mutual information decomposition rule \citep[Chap.~11]{NC}, rearranging the inequality and using the fact that $H(X)\leq \log|X|= \log d$ to maximise over the non-conditional entropies yields the relation. We will follow a similar approach.

We first show a generalisation of a bipartite quantum R\'enyi uncertainty relation found in~\cite[Eq.~(7.24)]{Fin}. One of the quantum R\'enyi decomposition rules from Theorem~\ref{result1} is then applied to derive a quantum R\'enyi information exclusion relation.
\subsection{A generalised bipartite quantum uncertainty relation}
We first establish a Maassen-Uffink type bipartite uncertainty relation expressed in terms of the generalised R\'enyi conditional entropy Eq.~\eqref{GCE}.
\begin{lemma}\label{GBUR}
Let $\Map_X\in\CPTP(A,X)$ and $\Map_Z\in \CPTP(A,Z)$ be two incompatible measurement maps, defined by the orthonormal basis $\{\ket{e_x}\}_x$ of $X$ such that
$\Map_X(\rho) = \sum_x\bra{e_x}\rho\ket{e_x}\ketbra{e_x}{e_x}$ and similarly for $\Map_Z$, $\{\ket{f_z}\}_z$ and $Z$.

For $\alpha, \beta, \gamma\geq 1/2$ such that $\frac{\alpha}{\alpha-1} = \frac{\beta}{1-\beta}+\frac{\gamma}{\gamma-1}$ and $(\alpha-1)(\beta-1)(\gamma-1) < 0$,
\begin{equation}\label{MUlike}
H_\beta(\Map_X(\rho_{AB})\|\sigma_B) + H_\gamma(\Map_Z(\rho_{AB})\|\sigma_B) \geq H_\alpha(\rho_{AB}\|\sigma_B)-\log c.
\end{equation}
\end{lemma}
Before we detail the proof of Lemma~\ref{GBUR} we first introduce a specific form of the Stinespring dilation~\cite{Stine}.
\begin{definition}[Stinespring dilation]
$\Map \in \CPTP(A, B)$ if and only if there exists an isometry $U\in \Lin(A, BC)$ such that $\Map(\rho) = \tr_C(U\rho U^\dagger) \text{ for all } \rho\in\Den(A)$.
\end{definition}
\begin{proof}[Proof of Lemma~\ref{GBUR}]
	Let $\mathcal{S}_Z\in \CPTP(A, ZZ')$ be the Stinespring dilation of $\Map_Z$ such that
	\begin{equation}
	\mathcal{S}_Z(\rho_A) = \sum_{z, z'}\bra{f_z}\rho_A\ket{f_{z'}}\ketbra{f_{z}}{f_{z'}}\otimes\ketbra{f_{z}}{f_{z'}}.
	\end{equation}
	We use the same argument as the proof of~\cite[Theorem~7.6]{Fin}, but without maximising over $\sigma_{Z'B}$, to arrive at
	\begin{equation}
	H_\alpha(\Map_X(\rho_{AB})\|\sigma_B)\geq H_\alpha(\mathcal{S}_Z(\rho_{AB})\|\sigma_{Z'B}) - \log c.\label{initineq}
	\end{equation}
	The two main components of this argument are the comparisons:
	\begin{align}
	\label{comp1}H_\alpha\left(\mathcal{S}_Z(\rho_{AB})\|\sigma_{Z'B}\right)&\leq -D_\alpha\left(\Map_X(\rho_{AB})\|\Map_X\left(\mathcal{S}_Z(\id_{Z}\otimes\sigma_{Z'B})\right)\right)\quad\text{and}\\
	\label{comp2}\Map_X\left(\mathcal{S}_Z(\id_{Z}\otimes\sigma_{Z'B})\right) &= \sum_{x,z}\left|\braket{e_x}{f_z}\right|^2\ketbra{e_x}{e_x}\otimes\bra{f_z}\sigma_{Z'B}\ket{f_z}\leq c\id_X\otimes\sigma_Z.
	\end{align}
	Substituting Eq.~\eqref{comp2} into Eq.~\eqref{comp1} yields Eq.~\eqref{initineq}.\\
	
	Let $\rho, \sigma\in \Den(ABC)$ be pure states and $\alpha, \beta, \gamma \geq 1/2$ such that ${\frac{\alpha}{\alpha-1} = \frac{\beta}{1-\beta}+\frac{\gamma}{\gamma-1}}$ and ${(\alpha-1)}{(\beta-1)}(\gamma-1) < 0$. Then by Theorem 1 in~\cite{Dup} we can write
	\begin{equation}
	H_\beta(\rho\|\sigma_{BC}) \geq H_\alpha(\rho\|\sigma_{C}) - H_\gamma(\rho_{BC}\|\sigma_C).\label{genchain}
	\end{equation}
	Substituting Eq.~\eqref{genchain} into Eq.~\eqref{initineq} we have
	\begin{equation}
	H_\beta(\Map_X(\rho_{AB})\|\sigma_B)\geq H_\alpha(\mathcal{S}_Z(\rho_{AB})\|\sigma_{B}) - H_\gamma(\tr_Z(\mathcal{S}_Z(\rho_{AB}))\|\sigma_B)- \log c.
	\end{equation}
	Using the fact that the marginals on $ZB$ and $ZB'$ of the state $\mathcal{S}_Z(\rho_{AB})$ are equivalent  and that the conditional entropies are invariant under local isometries we obtain Eq.~\eqref{MUlike}.
\end{proof}
\subsection{Applying the decomposition rule}
\begin{proof}[Proof of Theorem~\ref{result2}]
Starting with Eq.~\eqref{MUlike}, choosing parameters which satisfy the conditions and setting $\sigma = \rho$, we can write
\begin{equation}
H^\downarrow_{\bar\beta}(X|B)_\rho + H^\downarrow_{\bar\gamma}(Z|B)_\rho \geq -\log c + H^\downarrow_\alpha(A|B)_\rho.
\end{equation}
For each conditional entropy on the left-hand side we can derive the following inequalities from Eq.~\eqref{res1-3}:
\begin{align}
H^\downarrow_{\bar\beta}(X|B)_{\rho}&\leq H_{\tilde{\beta}}(X)_\rho-I^\uparrow_\beta(B\;;\>\!X)_\rho,\\
H^\downarrow_{\bar\gamma}(Z|B)_{\rho}&\leq H_{\tilde\gamma}(Z)_\rho-I^\uparrow_\gamma(B\;;\>\!Z)_\rho.
\end{align}
We can then write
\begin{align}
H_{\tilde{\beta}}(X)_\rho-I^\uparrow_\beta(B\;;\>\!X)_\rho + H_{\tilde\gamma}(Z)_\rho-I^\uparrow_\gamma(B\;;\>\!Z)_\rho &\geq -\log c + H^\downarrow_\alpha(A|B)_\rho\\
\implies I^\uparrow_\beta(B\;;\>\!X)_\rho + I^\uparrow_\gamma(B\;;\>\!Z)_\rho &\leq H_{\tilde{\beta}}(X)_\rho+H_{\tilde\gamma}(Z)_\rho + \log c -H^\downarrow_\alpha(A|B)_\rho\\
&\leq \log(d^2c) - H^\downarrow_\alpha(A|B)_\rho.
\end{align}
The last line is due to $H_{\alpha}(A)_\rho\leq \log d$ for all $\alpha$.

We can optimise the parameters when $\tilde\beta, \tilde\gamma=\frac{1}{2}$, hence
\begin{align}
\frac{\alpha}{\alpha-1} &\leq \frac{\beta}{\beta-1}+\frac{\tilde\beta}{\tilde\beta-1}+\frac{\gamma}{\gamma-1}+\frac{\tilde\gamma}{\tilde\gamma-1}\\
&\leq \frac{\beta}{\beta-1}-1+\frac{\gamma}{\gamma-1}-1\\
&\leq \frac{1}{\beta-1}+\frac{1}{\gamma-1}.
\end{align}
\end{proof}
Choosing $\alpha \rightarrow\infty$ we have the following corollary which summarises the possible choices of parameters which produce an optimal inequality.
\begin{corollary}\label{res2c}
	Given the same conditions as Theorem~\ref{result2}, for $\alpha\geq1/2$, we have
	\begin{equation}
	I^\uparrow_\alpha(B\;;\>\!X)_\rho + I^\uparrow_\frac{2\alpha - 3}{\alpha-2}(B\;;\>\!Z)_\rho \leq \log(d^2c)-H_{\min}(A|B)_\rho.
	\end{equation}
\end{corollary}
\begin{proof}
We know$\displaystyle{\lim_{\eta\rightarrow\infty} \eta/(\eta-1) = 1}$, hence if we take the order parameter on the conditional entropy to $\infty$, we obtain the relationship
\begin{align}
1 &\leq \frac{1}{\alpha-1} + \frac{1}{\beta-1},
\end{align}
from which we can deduce $\beta \leq \frac{\alpha-1}{\alpha-2} +1=\frac{2\alpha - 3}{\alpha-2}$.
\end{proof}
\appendix
\section{Analysis of related R\'enyi orders}\label{res}
The following lemma serves to explore the important relationship between the R\'enyi orders which is motivated by the application of Theorem~\ref{Bei}. We show what ranges result for each possible permutation of the signs of the orders by examining the asymptotic behaviour of this relationship. This then informs the possible choices of $\theta$ in Propositions~\ref{part1},~\ref{part2} and~\ref{part3}.
\begin{lemma}\label{abylem}
If $\alpha>0,\beta,\gamma>1/2$ and are related by
\begin{align}
\label{rel}\frac{\alpha}{\alpha-1} = \frac{\beta}{\beta-1} + \frac{\gamma}{\gamma-1}
\end{align}
and assuming, without loss of generality, that $\beta>\gamma$, then the following are true and cover all possible cases up to symmetry:
\begin{align*}
&\text{If }0<\frac{\alpha'}{\beta'}<1\text{ then either}\\
\numberthis&\qquad\text{\bf Case 1. }\alpha,\beta, \gamma>1,\quad\alpha<\gamma<\beta\quad\text{and}\quad\alpha\in(1, 2),\beta,\gamma\in(1,\infty),\label{case1}\\
&\text{or}\\
\numberthis&\qquad\text{\bf Case 2. }\alpha,\beta, \gamma<1,\quad \gamma<\beta<\alpha\quad\text{and}\quad\alpha\in[2/3, 1),\beta,\gamma\in[1/2,1).\label{case2}\\
&\text{If }0<\frac{\beta'}{\alpha'}<1\text{ then}\\
\numberthis&\qquad\text{\bf Case 3. }\alpha,\beta>1,\gamma<1,\quad\gamma<\beta<\alpha\quad\text{and}\quad\gamma\in[1/2,1), \beta \in (1,2),\alpha\in(1,\infty).\label{case3}\\
&\text{If }0<\frac{\gamma'}{\alpha'}<1\text{ then}\\
\numberthis&\qquad\text{\bf Case 4. }\alpha,\gamma<1,\beta>1,\quad\alpha<\gamma<\beta\quad\text{and}\quad\alpha\in(0,1), \gamma \in [1/2,1),\beta\in(1,\infty).\label{case4}
\end{align*}	
\end{lemma}
\begin{proof}
First we will investigate the possible cases or, more specifically, the cases missing from the lemma.
Given three independent binary options there are 8 possible permutations. Of the four that are missing the following:
$(\alpha, \gamma>1, \beta<1)$ and 
$(\alpha, \beta <1, \gamma>1)$, contradict the assumption that $\beta>\gamma$.
The remaining two:
$(\alpha>1, \beta,\gamma<1)$ and 
$(\alpha<1, \beta, \gamma>1)$,
never satisfy Eq.~\eqref{rel}. We can now explore the implications of each of the assumptions.

Consider $0<\frac{\alpha'}{\beta'}<1$. It is evident that $(\alpha - 1)(\beta - 1)>0$, a condition which now excludes \textbf{Case 4}. However, we can examine the two situations where this condition is satisfied:
\begin{align}
0<\frac{\alpha'}{\beta'}<1\implies \begin{cases}
\alpha<\beta \quad \text{if}\quad \alpha, \beta > 1\\
\alpha>\beta \quad \text{if}\quad \alpha, \beta < 1.
\end{cases}
\end{align}
It is clear that \textbf{Case 3} does not satisfy these implications but that \textbf{Cases 1}  and \textbf{2} do depending on the sign of $\alpha - 1$.

For \textbf{Case 1}, we can calculate that
$\displaystyle{\lim_{\eta\rightarrow 1^{+}} \frac{\eta}{\eta-1} =\infty}$ and
$\displaystyle{\lim_{\eta\rightarrow \infty} \frac{\eta}{\eta-1} = 1}$.

Since $\alpha, \beta$ and $\gamma$ are related by Eq.~\eqref{rel}, we have
\begin{align}
\alpha &\longrightarrow 1 \implies \beta,\gamma\longrightarrow 1\quad \text{and}\\
\alpha &\longrightarrow 2 \implies \beta,\gamma\longrightarrow \infty,
\end{align}
i.e. $1<\alpha<2$ and $1<\beta,\gamma<\infty$.

Moreover, for \textbf{Case 2}, another simple calculation shows that
$\displaystyle{\max_{1/2\leq\eta<1}\frac{\eta}{\eta-1} = -1}$\\
and $\displaystyle{\lim_{\eta\rightarrow 1^{-}} \frac{\eta}{\eta-1} =-\infty}$. Hence, $\alpha = \frac{2}{3} \implies \beta,\gamma = \frac{1}{2}$, i.e. $\frac{2}{3}\leq\alpha<1$ and $\frac{1}{2}\leq \beta,\gamma<1$.

If instead $0<\frac{\beta'}{\alpha'}<1$, we still have the condition $(\alpha - 1)(\beta - 1)>0$ but in the second part of the argument the inequalities are reversed, i.e
\begin{align}
0<\frac{\beta'}{\alpha'}<1\implies \begin{cases}
\alpha>\beta \quad \text{if}\quad \alpha, \beta > 1\\
\alpha<\beta \quad \text{if}\quad \alpha, \beta < 1.
\end{cases}
\end{align}
This overall excludes \textbf{Cases 1}, \textbf{2}  and \textbf{4} but satisfies \textbf{Case 3}.

In this situation we again have $\alpha\rightarrow 1 \implies \beta, \gamma\rightarrow 1$ and for fixed $\gamma$ we can write $\displaystyle{\lim_{\alpha\rightarrow\infty}\beta = \frac{1}{\gamma}}$. Given that $\gamma>1/2$, this implies $1<\beta<2$.

Lastly, we have $0<\frac{\gamma'}{\alpha'}<1$, which implies $(\alpha - 1)(\gamma-1)>0$, excluding \textbf{Case 3} Similarly, we have following situations:
\begin{align}
0<\frac{\gamma'}{\alpha'}<1\implies \begin{cases}
\alpha>\gamma \quad \text{if}\quad \alpha, \gamma > 1\\
\alpha<\gamma \quad \text{if}\quad \alpha, \gamma < 1,
\end{cases}
\end{align}
which exclude \textbf{Cases 1} and \textbf{2}. So \textbf{Case 4} is the only remaining case which is satisfied.

We again have $\alpha\rightarrow 1 \implies \beta, \gamma\rightarrow 1$ and for fixed $\gamma$, $\displaystyle{\lim_{\alpha\rightarrow 0} \beta = \frac{\gamma}{2\gamma-1}}$ and $\displaystyle{\lim_{\gamma\rightarrow \frac{1}{2}^+}\frac{\gamma}{2\gamma-1} = \infty}$. Hence $1<\beta<\infty.$
\end{proof}
We conclude with the following useful corollary
\begin{corollary}\label{shyeah}
Given the assumptions in Lemma~\ref{abylem}  we have that 
\begin{align}
\alpha<\gamma<\beta &\implies (\alpha-1)(\beta-1)(\gamma-1)>0\quad\text{and}\\
\gamma<\beta<\alpha&\implies(\alpha-1)(\beta-1)(\gamma-1)<0.
\end{align}
\end{corollary}
\begin{proof}
This is evident from examining each case of Lemma~\ref{abylem}.
\end{proof}

\bibliography{LRbib,libraryMT}
\end{document}